\documentclass[conference]{IEEEtran}
\usepackage{algorithm,amsbsy,amsmath,amssymb,epsfig,bbm,mathrsfs,multirow,amsthm}
\usepackage{color}
\usepackage{algpseudocode}
\usepackage{graphicx,caption,subcaption,array,multirow,bm}
\usepackage{makecell}
\usepackage{threeparttable}

\DeclareMathAlphabet{\mathpzc}{OT1}{pzc}{m}{it}

\newtheorem{theorem}{\textbf{\textsc{Theorem}}}

\begin{document}

\title{Defeating Eavesdroppers with Ambient Backscatter Communications}

\author{ Nguyen Van Huynh$^1$, Nguyen Quang Hieu$^2$, Nam H. Chu$^2$, Diep N. Nguyen$^2$,\\
	Dinh Thai Hoang$^2$, and Eryk Dutkiewicz$^2$\\
	$^1$ Department of Electrical and Electronic Engineering, Imperial College London, UK\\
	$^2$ School of Electrical and Data Engineering, University of Technology Sydney, Australia\\
	Emails: huynh.nguyen@imperial.ac.uk, \{hieu.nguyen-1, namhoai.chu\}@student.uts.edu.au,\\
	 \{diep.nguyen, hoang.dinh, eryk.dutkiewicz\}@uts.edu.au
}

\maketitle

\begin{abstract}
Unlike conventional anti-eavesdropping methods that always require additional energy or computing resources (e.g., in friendly jamming and cryptography-based solutions), this work proposes a novel anti-eavesdropping solution that comes with mostly no extra power nor computing resource requirement. This is achieved by leveraging the ambient backscatter technology in which secret information can be transmitted by backscattering it over ambient radio signals. Specifically, the original message at the transmitter is first encoded into two parts: (i) active transmit message and (ii) backscatter message. The active transmit message is then transmitted by using the conventional wireless transmission method while the backscatter message is transmitted by backscattering it on the active transmit signals via an ambient backscatter tag. As the backscatter tag does not generate any active RF signals, it is intractable for the eavesdropper to detect the backscatter message. Therefore, secret information, e.g., a secret key for decryption, can be carried by the backscattered message, making the adversary unable to decode the original message. Simulation results demonstrate that our proposed solution can significantly enhance security protection for communication systems.

\end{abstract}



\section{Introduction}
\label{Sec:intro}

\IEEEPARstart{T}{he} provisioning of security and privacy has been emerging as a critical issue in wireless communications systems due to the broadcast nature of the wireless medium. Among security threats, eavesdropping attacks are often considered as the most popular threats in wireless communication systems, especially in IoT networks. In particular, to perform eavesdropping attacks, an eavesdropper is usually placed near the target system to ``wiretap'' the legitimate channel and obtain the information sent from the transmitter. As the eavesdropper works in a passive manner, it is very challenging for the legitimate system to detect and prevent such eavesdropping attacks.

Anti-eavesdropping has been well investigated in the literature, e.g.,~\cite{Soltani2018Covert}-\cite{Zhang2016On}. The most common approach is to rely on ``friendly jamming'' in which interference is deliberately injected into the channel to disrupt the signal reception at potential eavesdroppers, e.g.,~\cite{Soltani2018Covert, Siyari2017Frinedly, Mobini2019Wireless}. However, friendly jamming cannot always guarantee a positive secrecy rate, defined as the difference between the channel capacity between the transmitter and the legitimate receiver and that between the transmitter and the eavesdropper. Moreover, generating artificial noise may degrade legitimate signal reception at nearby legitimate devices, especially in dense wireless settings which are very common in future wireless networks. Recently, cooperative transmission by using relays has been emerging as a promising technique to improve the physical-layer security of wireless communications under the presence of eavesdroppers~\cite{Yang2017Optimal}. These relays can also generate jamming signals (e.g., noise) to obfuscate eavesdroppers in the ranges of relays. Nevertheless, this approach usually requires eavesdropping channel state information in advance to achieve good protection performance~\cite{Yang2017Optimal}. In practice, eavesdropping channel state information is usually unavailable or difficult to accurately estimate due to the passive nature of eavesdroppers. In addition, additional relays also come at the cost of higher complexity. In practice, a more popular and acceptable solution to deal with eavesdroppers is to encrypt the information at the application and transportation layers~\cite{Hu2018Covert}. However, the encrypted data can be decrypted if the eavesdropper has sufficient computational capacities. In addition, distributing and managing cryptographic keys are challenging, especially in decentralized systems with a massive number of devices and mobility~\cite{Lu2020Intelligent}. Moreover, for power-constrained devices such as IoT transceivers (e.g., smart meters) it is difficult, if not infeasible, to effectively run computation-demanding cryptographic functions~\cite{Zhang2016On}. Note that all aforementioned anti-eavesdropping methods always require significant additional energy or computing resources (e.g., in friendly jamming and cryptographic solutions).

Given the above, we propose a novel anti-eavesdropping solution that comes at mostly no extra power nor computing resource cost. This is achieved by augmenting legitimate transmitters with an ambient backscatter tag that can backscatter information bits onto ambient radio signals \cite{ABC}. Specifically, the original message is first encoded into two parts: (i) active transmit message and (ii) backscatter message. The active transmit message is sent to the receiver using conventional (active) wireless signals from the transmitter. At the same time, the backscatter message (the second part of the original message) is backscattered to the receiver by using an ambient backscatter tag. Note that the ambient backscatter message is transmitted at the same time and on the same frequency as the active transmit message, yet at no extra transmission power, hence can be considered as pseudo noise in the background~\cite{Huynh2018Survey, Liu2013Ambient}. For that, it is mostly impossible for the eavesdropper to discern/eavesdrop the ambient message. Hence, the backscatter message is used to carry secret information, e.g., a secret key for decryption. In this case, even if the eavesdropper can obtain the information from the active signals, it still cannot decode the original message (due to missing the important information sent over the backscatter signals). Simulation results show that our proposed approach can effectively defeat the eavesdropper. More importantly, the ambient backscatter tag can operate without requiring any power supply~\cite{Huynh2018Survey}. As such, our proposed solution can enable green and secure communications for many applications in future wireless communication networks. For example, a smart meter can be equipped with the proposed backscatter tag to transmit secret messages (e.g., key) over the ambient backscatter channel.

\section{System Model}
\label{Sec.System}
\begin{figure}[h]
	\centering
	\includegraphics[width=0.95\linewidth]{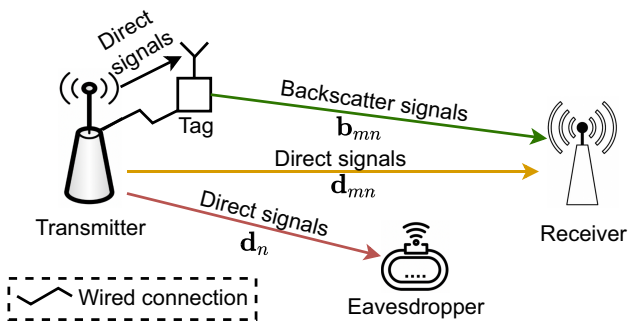}
	\caption{System model.}
	\label{Fig.system_model}
\end{figure}

We consider a wireless network consisting of a transmitter that is aided with a nearby ambient backscatter tag and a legitimate receiver with the presence of an eavesdropper, as illustrated in Fig.~\ref{Fig.system_model}. The ambient backscatter tag is equipped with a backscatter circuit and connected to the transmitter through a wired channel. When the transmitter sends a message to the receiver, it will first encode this message into two parts as illustrated in Fig.~\ref{Fig.splitmessage}. The first part will be transmitted to the receiver over the conventional communication channel ($d_\mathrm{mn}$) based on the active RF component of the transmitter. The second part will be simultaneously transmitted to the receiver through the backscatter tag. It is important to note that, when the transmitter transmits signals over the conventional communication channel, the tag will backscatter such signals and transmit the data for the second part to the receiver. As a result, at the same time, we can transmit two data streams (one is over the conventional channel and another one is over a hidden channel, i.e., a backscatter channel) to the receiver. It is also worth mentioning that the backscatter rate is lower than the active transmission rate. Thus, the size of the second part is usually smaller than the first part. As the backscatter tag does not generate any active signals, it is intractable for the eavesdropper to detect the backscattered signals. As a result, the eavesdropper may not be able to derive the original message, thus the deception strategy for communication between the transmitter and the receiver can be guaranteed.

Details of splitting the original message are illustrated in Fig.~\ref{Fig.splitmessage}. In particular, we randomly take a number of bits from the original message to construct the backscatter message with a step of $K$ symbols. The rest of the original message is conveyed to the receiver by the transmitter through active transmissions. The size of the active transmit message is usually larger than that of the backscatter message. In this paper, we assume that the backscatter frame has the size of $I$ bits in which $P$ bits are reserved for pilot signals and $S$ bits are used for dividing information (i.e., the first bit's ID and the step $K$) with $(P+S) < I$. In this way, we can significantly improve the security level of the system as the eavesdropper cannot derive the backscatter message as well as the divided information.

\begin{figure}[h]
	\centering
	\includegraphics[width=0.95\linewidth]{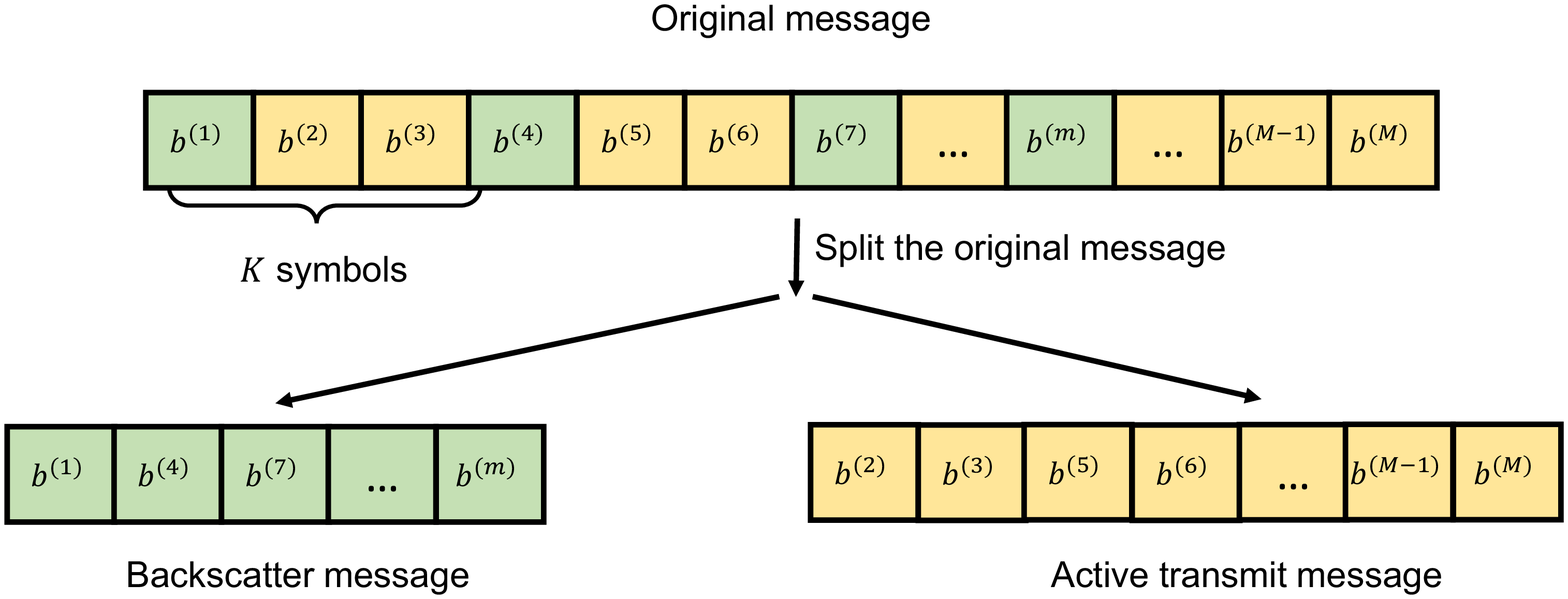}
	\caption{Dividing the original message.}
	\label{Fig.splitmessage}
\end{figure}

\section{Channel Model}
\label{Sec.signalmodel}

Practically, having the backscatter rate lower than the sampling rate of the transmitter's signals will ensure the receiver properly decodes the backscattered signals~\cite{Huynh2018Survey, Liu2013Ambient, Huynh2022Defeating}. Thus, in this paper, we assume that the sampling rate of the transmitter's signals is $N$ times higher than the backscatter rate. To do that, the ambient backscatter tag will backscatter each information bit over $N$ transmitter symbols. We then denote $y_{mn}$ as the $n$-th signal received at the $m$-th antenna of the receiver. In our system, $y_{mn}$ includes the signals sent from the transmitter, the backscattered signals from the ambient backscatter tag, and noise. As such, we have
\begin{equation}
	\label{eq:totalreceive}
	y_{mn} = \underbrace{d_{mn}}_{\text{direct link}} + \underbrace{b_{mn}}_{\text{backscatter link}} + \sigma_{mn},
\end{equation}
where $d_{mn}$ is the transmitter's direct link signals, $b_{mn}$ presents the backscattered signals, and $\sigma_{mn}$ denotes the circularly symmetric complex Gaussian (CSCG) noise with unit variance and zero mean (i.e., $\sigma_{mn} \sim \mathcal{CN}(0,1)$).

\subsection{Direct Link}

We denote $s_{tn}$ as the transmitter's signals at time instant $n$. It is worth noting that the tag considers the transmitter's signals as ambient signals and backscatters these signals to convey information to the receiver. Thus, we assume that $s_{tn}$ is random and unknown at the backscatter tag and follows the standard CSCG distribution with zero mean and unit variance (i.e., $s_{tn} \sim \mathcal{CN}(0,1)$). The signals from the transmitter received at the $m$-th antenna of the receiver can be expressed as follows:

\begin{equation}
	d_{mn} = f_{rm} \sqrt{P_{tr}}s_{tn}, 
\end{equation}
where $f_{rm} $ is the Rayleigh fading with $\mathbb{E}[|f_{rm} |^2] = 1$~\cite{Huynh2022Defeating}. $P_{tr}$ denotes the average received powered from the transmitter. $P_{tr}$ can be expressed as
\begin{equation}
	P_{tr} = \frac{ \kappa P_t G_t G_r}{ {L_r}^{\upsilon}},
\end{equation}
where $\kappa = {(\frac{\lambda}{4\pi})}^2$ with wavelength $\lambda$. $P_t$ denotes the transmitter's transmit power. $G_t$ and $G_r$ are the antenna gains of the transmitter and the receiver, respectively. $L_r$ is the transmitter-to-receiver distance. $\upsilon$ is the path loss exponent.

\subsection{Backscatter Link}
The transmitter's signals received at the backscatter tag can be expressed as follows:

\begin{equation}
	c_{n} = g_r \sqrt{P_{b}}s_{tn},
\end{equation}
where $g_r$ denotes the Rayleigh fading from the transmitter to the tag with $\mathbb{E}[|g_r|^2] = 1$~\cite{Huynh2022Defeating}. $P_{b}$ is the average power from the transmitter received at the tag. We have
\begin{equation}
	P_{b} = \frac{\kappa P_t G_t G_b}{{L_b}^{\upsilon}},
\end{equation}
where $G_b$ is the antenna gain at the tag and $L_b$ is the transmitter-to-tag distance. The tag backscatters information to the receiver by reflecting or absorbing the transmitter's signals. We denote the reflecting state as $e=1$ and the absorbing state as $e=0$. As mentioned, each information bit will be backscattered over $N$ transmitter symbols. Thus, state $e$ remains unchanged during this period. We then can express the backscattered signals as follows:

\begin{equation}
	s_{b,n} = \gamma c_n e,
\end{equation}
where $\gamma$ is the reflection coefficient. The backscattered signals received at the $m$-th antenna of the receiver can be expressed as follows:

\begin{equation}
	\label{eq:bn}
	\begin{aligned}
		b_{mn} &= f_{bm} \sqrt{\frac{G_b G_r \kappa}{{L_e}^{\upsilon}}} \gamma e \Big(g_r \sqrt{P_{br}}s_{tn}\Big)\\
		&= f_{bm} e \Big(g_r \sqrt{\frac{\kappa |\gamma|^2 P_{tr}{G_b}^2{L_r}^{\upsilon}}{{L_b}^{\upsilon}{L_e}^{\upsilon}}} s_{tn}\Big),
	\end{aligned}
\end{equation}
where $L_e$ is the tag-to-receiver distance and $f_{bm}$ is the Rayleigh fading of the tag-to-receiver link with $\mathbb{E}[|f_{bm}|^2] = 1$. Denote $\tilde{\alpha}_r = \frac{\kappa {|\gamma|^2}{G_b}^2{L_r}^{\upsilon}} {{L_b}^{\upsilon} {L_e}^{\upsilon}}$, we can rewrite (\ref{eq:bn}) as

\begin{equation}
	b_{mn} = f_{bm} e \Big(g_r \sqrt{\tilde{\alpha}_r P_{tr}}s_{tn} \Big).
\end{equation}

\subsection{Received Signals}
\label{sec:received}
Given the above, the received signals at the $m$-th antenna then can be expressed as follows:
\begin{equation}
	y_{mn} = f_{rm} \sqrt{P_{tr}}s_{tn} + f_{bm} e \Big(g_r \sqrt{\tilde{\alpha}_r P_{tr}}s_{tn} \Big) + \sigma_{mn}.
\end{equation}
Denote $\alpha_{dt} \triangleq P_{tr} $ as the signal-to-noise ratio (SNR) of the transmitter-to-receiver link and $\alpha_{bt} \triangleq \tilde{\alpha}_r P_{tr}$ as the SNR of the backscatter link (i.e., transmitter-tag-receiver link), we have
\begin{equation}
	y_{mn} = \underbrace{f_{rm} \sqrt{\alpha_{dt}}s_{tn}}_{\text{direct link}} + \underbrace{f_{bm}  e \Big(g_r \sqrt{\alpha_{bt}}s_{tn} \Big)}_{\text{backscatter link}} + \sigma_{mn}.
\end{equation}
Denote
\begin{equation}
	\begin{aligned}
	&\mathbf{f}_{r} = [f_{r1}, \ldots, f_{rm}, \ldots, f_{rM}]^\mathsf{T},\\
	&\mathbf{f}_{b} = [f_{b1}, \ldots, f_{bm}, \ldots, f_{bM}]^\mathsf{T},\\
	&\bm{\sigma}_{n} = [\sigma_{1n}, \ldots, \sigma_{mn}, \ldots, \sigma_{Mn}]^\mathsf{T},
	\end{aligned}
\end{equation}
we can express the total received signals at the receiver as follows:

\begin{equation}
	\begin{aligned}
		\mathbf{y}_{n} &= [y_{1n}, \ldots, y_{mn}, \ldots, y_{Mn}]^\mathsf{T} \\
		&= \underbrace{\mathbf{f}_{r} \sqrt{\alpha_{dt}}s_{tn}}_{\text{direct link}} + \underbrace{\mathbf{f}_{b}  e \Big(g_r \sqrt{\alpha_{bt}}s_{tn} \Big)}_{\text{backscatter link}} + \bm{\sigma}_{n},
	\end{aligned}
\end{equation}
In this work, each backscatter frame $\mathbf{b}$ contains $I$ information bits, denoted by $\mathbf{b} = [b^{(1)},\ldots, b^{(i)}, \ldots, b^{(I)}]$. The channel is assumed to be invariant during one backscatter frame. Each information bit is encoded before backscattering with the modulo-2 operation as follows:

\begin{equation}
	e^{(i)} = e^{(i-1)} \oplus b^{(i)},
\end{equation}
where $e^{(i)}$ is the encoded bits in which $e^{(0)} = 1$~\cite{Guo2019Nocoherent, Huynh2022Defeating} and $\oplus$ denotes the addition modulo 2. As mentioned, $e^{(i)}$ will be backscattered to the receiver over $N$ transmitter symbols. Thus, the received signals at the receiver during the $i$-th backscatter symbol period can be expressed by:
\begin{equation}
	\mathbf{y}_{n}^{(i)} = \mathbf{f}_{r} \sqrt{\alpha_{dt}}s_{tn}^{(i)} + \mathbf{f}_{b}  e^{(i)} \Big(g_r \sqrt{\alpha_{bt}}s_{tn}^{(i)} \Big) + \bm{\sigma}_{n}^{(i)},
\end{equation}
where $n=1,2,\ldots, N$ and $i=1,2,\ldots, I$. Note that with the ambient backscatter communication technology, the transmitter's signals at the tag are unknown and random. As such, it is impossible to derive the close-form of the backscatter rate $R_\text{b}$~\cite{Huynh2018Survey},~\cite{Guo2019Nocoherent}. Instead, in Theorem~\ref{theo:maxrate}, we obtain the maximum achievable backscatter rate $R_\text{b}^\dagger$ to evaluate the system performance.

\begin{theorem}
	\label{theo:maxrate}
	The maximum achievable backscatter rate $R_\text{b}^\dagger$ of the backscatter tag can be numerically obtained as follows:
	\begin{equation}
		\begin{aligned}
			R_\text{b}^\dagger &= C(\theta_0) - \mathbb{E}_{\mathbf{y}_0}[C(\omega)] \\
			&= C(\theta_0) - \int_{\mathbf{y}_0}(\theta_0p(\mathbf{y}_0|e=0) +\theta_1 p(\mathbf{y}_0|e=1))C(\omega_0)d\mathbf{y}_0,
		\end{aligned}
	\end{equation}
	where $\theta_0$ is the prior probability when backscattering bits 0 (the prior probability when backscattering bits 1 is $\theta_1 = 1-\theta_0$), $\mathbf{y}_0$ is a realization of $\mathbf{y}$, $C$ is the binary entropy function, and $\omega_j = p(e=j|\mathbf{y}_0), j \in \{0,1\}$ is the posterior probability of backscattered bit $e$ given the received signal $\mathbf{y}_0$.
\end{theorem}
\begin{proof}
	The proof of Theorem~\ref{theo:maxrate} is provided in Appendix~\ref{appendix:maximumbackscatterrate}.
\end{proof}

\section{Decoding Backscattered Signals with Maximum Likelihood Detector}
\label{sec:decoding}
\subsection{Maximum Likelihood Detector}

Practically, the backscattered signals from the ambient backscatter tag are often weaker than the active signals sent from the transmitter. As such, it is very challenging to detect the backscattered signals. In the following, we propose an optimal maximum likelihood (ML) detector to help the receiver decode the backscattered signals.

When the tag backscatters bits ``0'' (i.e., $e^{(i)} = 0$), the received signals contain only the direct link signals. Differently, when backscattering bits ``1'' (i.e., $e^{(i)} = 1$), the received signals contain both the active signals from the transmitter and the backscattered signals from the tag. As such, the channel statistical covariance matrices corresponding to these cases can be expressed as~\cite{Guo2019Nocoherent},~\cite{Huynh2022Defeating}

\begin{equation}
	\begin{aligned}
	&\mathbf{K}_0 = \mathbf{h}_1 \mathbf{h}_1^\text{H} + \mathbf{I}_M,\\
	&\mathbf{K}_1 = (\mathbf{h}_1+\mathbf{h}_2)(\mathbf{h}_1+\mathbf{h}_2)^\text{H}+ \mathbf{I}_M,
	\end{aligned}
\end{equation}
where $\mathbf{h}_1 = \mathbf{f}_r \sqrt{\alpha_{dt}}$, $\mathbf{h}_2 =  g_r \mathbf{f}_b \sqrt{\alpha_{bt}}$, $\mathbf{I}_M$ is the $M \times M$ identity matrix, and $(*)^\text{H}$ denotes the conjugate transpose operator. Given received signals $\mathbf{y}_n^{(i)}$ and  backscatter symbol $e^{(i)}$, the conditional probability density functions (PDFs) are then obtained as follows:

\begin{equation}
	\label{eq:pdf}
	\begin{aligned}
		p(\mathbf{y}_n^{(i)}|e^{(i)} = 0) &= \frac{1}{\pi^M |K_0|} e^{-{\mathbf{y}_n^{(i)}}^\text{H} K_0^{-1}\mathbf{y}_n^{(i)}}, \\
		p(\mathbf{y}_n^{(i)}|e^{(i)} = 1) &= \frac{1}{\pi^M |K_1|} e^{-{\mathbf{y}_n^{(i)}}^\text{H} K_1^{-1}\mathbf{y}_n^{(i)}}.
	\end{aligned}
\end{equation}
From (\ref{eq:pdf}), the likelihood functions of $\mathbf{Y}^{(i)} = [\mathbf{y}_1^{(i)}, \ldots, \mathbf{y}_n^{(i)}, \ldots, \mathbf{y}_N^{(i)}]^\mathsf{T}$ can be calculated as follows~\cite{Guo2019Nocoherent, Huynh2022Defeating}:

\begin{equation}
	\label{eq:likelihoodfunctions}
	\begin{aligned}
		\mathcal{L}(\mathbf{Y}^{(i)}|e^{(i)} = 0) &= \prod_{n=1}^{N} \frac{1}{\pi^M |K_0|} e^{-{\mathbf{y}_n^{(i)}}^\text{H} K_0^{-1}\mathbf{y}_n^{(i)}},\\
		\mathcal{L}(\mathbf{Y}^{(i)}|e^{(i)} = 1) &= \prod_{n=1}^{N} \frac{1}{\pi^M |K_1|} e^{-{\mathbf{y}_n^{(i)}}^\text{H} K_1^{-1}\mathbf{y}_n^{(i)}}.
	\end{aligned}
\end{equation}
Then, we can derive the ML criterion (i.e., hypothesis) for backscattered symbol $e^{(i)}$ as follows:

\begin{equation}
	\hat{e}^{(i)} 	=	\left\{	\begin{array}{ll}
		0,	&	\mathcal{L}(\mathbf{Y}^{(i)}|e^{(i)} = 0) > \mathcal{L}(\mathbf{Y}^{(i)}|e^{(i)} = 1),\\
		1,	&	\mathcal{L}(\mathbf{Y}^{(i)}|e^{(i)} = 0) < \mathcal{L}(\mathbf{Y}^{(i)}|e^{(i)} = 1),
	\end{array}	\right.
\end{equation}
where $\hat{e}^{(i)}$ denotes the estimated bit. The ML criterion then can be rewritten as follows~\cite{Huynh2022Defeating}:

\begin{equation}
	\hat{e}^{(i)}	=	\left\{	\begin{array}{ll}
		0,	&	\sum_{n=1}^{N} {\mathbf{y}_n^{(i)}}^\text{H} (K_0^{-1}-K_1^{-1})\mathbf{y}_n^{(i)} < N\ln\frac{|K_1|}{|K_0|},\\
		1,	&	\sum_{n=1}^{N} {\mathbf{y}_n^{(i)}}^\text{H} (K_0^{-1}-K_1^{-1})\mathbf{y}_n^{(i)} > N\ln\frac{|K_1|}{|K_0|}.
	\end{array}	\right.
\end{equation}
Based on $\hat{e}^{(i)}$ we can derive the backscattered bit $e^{(i)}$ and then recover the original bit $b^{(i)}$.

\subsection{Successfully Decoded Information at the Receiver}
Let $\epsilon_d$ and $\epsilon_b$ denote the bit error ratio (BER) of the signals transmitted over the direct link and backscatter link, respectively. Given $\epsilon_d$ and $\epsilon_b$, the number of successfully decoded bits at the receiver, denoted as $\mathbf{\bar{T}}$, can be expressed as

\begin{equation}
	\mathbf{\bar{T}} = \mathbf{T} (1 - \eta) (1 - \epsilon_d) + \mathbf{T} \eta (1 - \epsilon_b),
\end{equation}
where $\mathbf{T}$ is the total number of bits transmitted from the transmitter, $\eta \in [0, 1]$ is the splitting ratio between backscatter bits and direct-transmission bits. For example, given $\mathbf{T} = 1,000$ bits, $\eta = 0.1$ splits 1,000 bits into 900 bits and 100 bits to be transmitted over direct link and backscatter link, respectively. Here, $\mathbf{\bar{T}}$ is calculated as a sum of the number of bits successfully transmitted over the direct link, i.e., $\mathbf{T} (1 - \eta) (1 - \epsilon_d)$, and the number of bits successfully transmitted over the backscatter link, i.e., $\mathbf{T} \eta (1 - \epsilon_b)$.


\section{Simulation Results}
\label{sec:evaluation}
\subsection{Parameter Setting}
In this section, we evaluate the performance of our proposed solution in various scenarios. Unless otherwise stated, we set $N=5$ and $I=100$. The number of antennas at the receiver is varied from 1 to 10. The Rayleigh fading follows the standard CSCG distribution with unit variance and zero mean~\cite{Zhang2019Constellation},~\cite{Huynh2022Defeating}. It is worth mentioning that $\alpha_\mathrm{dt}$ greatly depends on the environment factors such as antenna gain, transmit power, transmitter-to-receiver distance, and path loss. Therefore, we vary $\alpha_\mathrm{dt}$ from 1dB to 9dB in the simulations to evaluate the system performance in different scenarios. As the backscattered signals are usually weak, we set $\tilde{\alpha}_\mathrm{r}$ at -10dB. To obtain robust and reliable results, all simulations in this section are averaged over $10^6$ Monte Carlo runs.

\subsection{Performance Evaluation}
\begin{figure}[!]
	\centering
	\begin{subfigure}[b]{0.4\textwidth}
		\centering
		\includegraphics[scale=0.43]{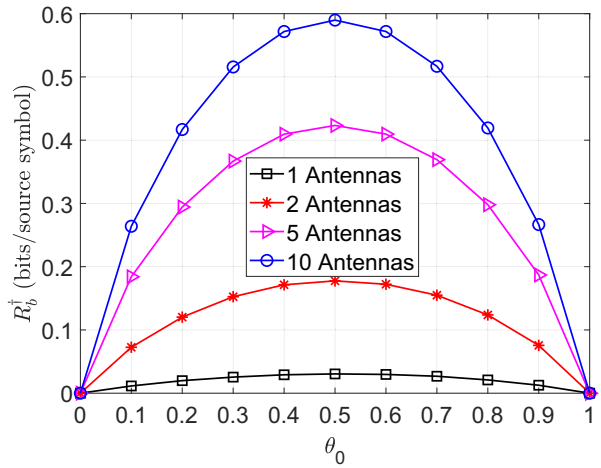}
		\caption{}
	\end{subfigure}%
	
	\begin{subfigure}[b]{0.4\textwidth}
		\centering
		\includegraphics[scale=0.43]{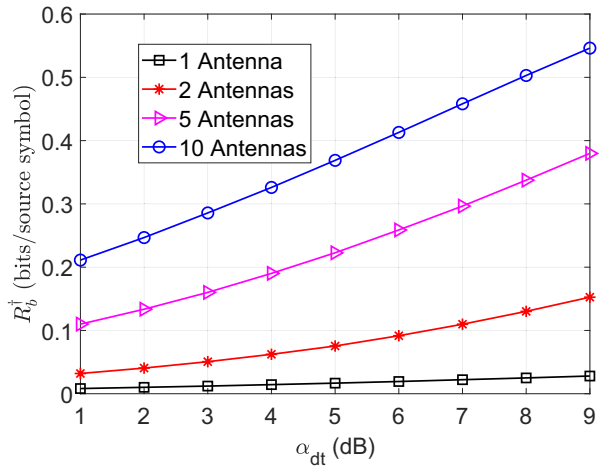}
		\caption{}
	\end{subfigure}%
	\caption{Maximum average achievable backscatter rate vs. (a) $\theta_0$ and (b) $\alpha_{dt}$.} 
	\label{fig.MaximumRate}
\end{figure}

We first vary the prior probability of backscattering bits ``$0$'' and observe the maximum achievable backscatter rate of the tag with different numbers of antennas at the receiver as illustrated in Fig.~\ref{fig.MaximumRate}(a). In particular, the maximum achievable backscatter rate is obtained based on Theorem~\ref{theo:maxrate} through $10^6$ Monte Carlo runs. It can be observed that the backscatter rate increases with the number of antennas at the receiver. The reason is that with multiple antennas, the receiver can leverage the antenna gain to eliminate the effects of the fading and the direct link interference. As a result, the backscattered signals received at the receiver can be enhanced. It is worth noting that when the probability of backscattering bits ``$0$'' equals 0.5, the backscatter rate is maximized. In Fig.~~\ref{fig.MaximumRate}(b), we vary $\alpha_\mathrm{dt}$ and observe $R_\text{b}^\dagger$. Clearly, when $\alpha_\mathrm{dt}$ increases, the achievable backscatter rate increases because the tag can backscatter strong signals to the receiver.

\begin{figure}[!]
	\centering
	\begin{subfigure}[b]{0.43\textwidth}
		\centering
		\includegraphics[width=1\linewidth]{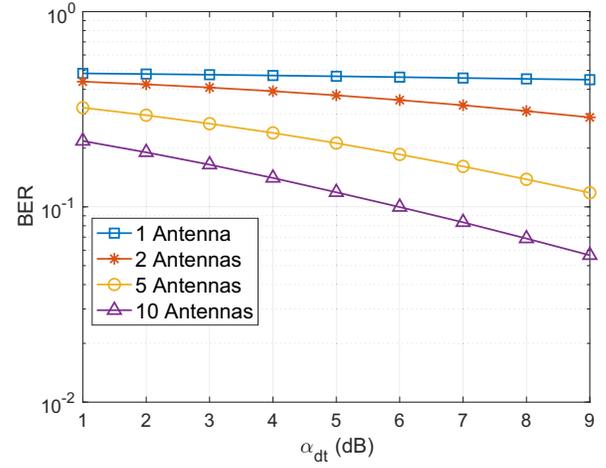}
		\caption{}
	\end{subfigure}%
	
	\begin{subfigure}[b]{0.43\textwidth}
		\centering
		\includegraphics[width=1\linewidth]{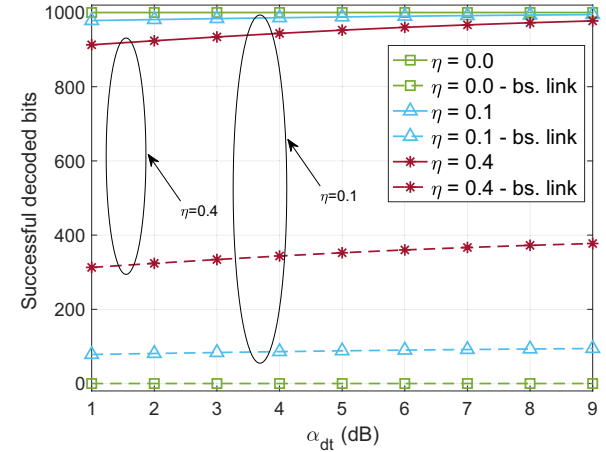}
		\caption{}
	\end{subfigure}%
	\caption{(a) BER $\epsilon_b$ vs. $\alpha_{dt}$ and (b) total number of successfully decoded bits $\mathbf{\bar{T}}$ vs. $\alpha_{dt}$.}
	\label{fig:BER}
\end{figure}

In Fig.~\ref{fig:BER}, we vary $\alpha_\mathrm{dt}$ from 1dB to 9dB and observe the BER of the system as well as the number of the successfully decoded bits. It is noted that we only consider the BER of the signals transmitted over the backscatter link, i.e., $\epsilon_b$, since the value of $\epsilon_d$ is approximate 0. It can be observed from Fig.~\ref{fig:BER}(a) that the BER performance increases with $\alpha_\mathrm{dt}$. The reason is that the backscattered signals received at the receiver can be improved when the tag backscatters strong signals from the transmitter. Moreover, the BER decreases when the number of antennas at the receiver increases. This is due to the fact that with more antennas, the receiver can reduce the effects of interference and fading by leveraging the antenna gain to strengthen the received backscattered signals.

In Fig.~\ref{fig:BER}(b), we consider the scenario with $M=10$ antennas at the receiver and vary the splitting ratio $\eta$ to evaluate the effectiveness of the backscatter tag. We consider that the total number of $\mathbf{T} = 1,000$ bits are transmitted over a period of time (e.g., a time frame). It is noted that the dashed lines express the number of successfully decoded bits of the backscatter link (bs. link) and the solid lines express the total number of bits successfully transmitted from the transmitter. The dashed lines, in other words, illustrate the amount of backscatter information that can be securely transmitted without being detected by the eavesdropper. It can be observed that with $\eta = 0$, i.e., no information is transmitted over the backscatter link, the successfully decoded signals achieve approximately 1,000 bits with the BER $\epsilon_d \approx 0$. When $\eta$ increases to 0.1 and 0.4, the gap between the respective solid lines and dash lines shrinks. The reason is that the more signals are transmitted over the backscatter link, the more signals are lost because the BER of the backscatter link is much higher than that of the direct link, i.e., $\epsilon_b \gg \epsilon_d$. The results express the trade-off between the number of bits that can be hidden from the eavesdropper and the number of successful received bits at the receiver. Given the above, our proposed solution is very promising in dealing with the eavesdropper for the following reasons. First, backscatter communications are easy to implement in practice. Second, due to a new way of communication, it introduces more difficulties to attackers in decoding the actual information.

\section{Conclusion}
\label{sec:conclusion}
In this paper, we have introduced a novel anti-eavesdropping solution that comes with mostly no extra power nor computing resource requirement by using the ambient backscatter communication technology. Specifically, a part of the original message will be sent over conventional active transmissions to attract the eavesdropper and drain its energy. The other part will be transmitted to the receiver by backscattering the transmitter's signals. In this way, it is impossible for the eavesdropper to obtain and derive the original message. The analytical and simulation results have demonstrated the effectiveness of our proposed solution in dealing with the eavesdropper.
\appendices
\section{The proof of Theorem~\ref{theo:maxrate}}
\label{appendix:maximumbackscatterrate}
In the following, we will mathematically present how to obtain the maximum achievable backscatter rate of the backscatter tag (similar to~\cite{Guo2019Nocoherent}). First, it can be observed that, $R_\text{b} = I(e;\mathbf{y})$ is the mutual information between the modulated information $e$ and the received signals $\mathbf{y}$ at the receiver. Hence, the maximum achievable backscatter rate $R_\text{b}^\dagger$ can be expressed as follows~\cite{Guo2019Nocoherent}:

\begin{equation}
	\label{eq:Rb}
	R_\text{b}^\dagger = \mathbb{E}[I(e, \mathbf{y})],
\end{equation}
where the mutual information $I(e, \mathbf{y})$ is formulated as follows:

\begin{equation}
	I(e, \mathbf{y}) = C(\theta_0) - \mathbb{E}_{\{\mathbf{y}_0\}}[H(e|\mathbf{y}_0)],
\end{equation}
where $H(e|\mathbf{y}_0)$ is the conditional entropy of $e$ given $\mathbf{y}_0$, and $C(\theta_0)$ denotes the binary entropy function which is calculated in (\ref{eq:Ctheta0}).

\begin{equation}
	\label{eq:Ctheta0}
	C(\theta_0) \triangleq -\theta_0 \log_2 \theta_0 -\theta_1 \log_2 \theta_1.
\end{equation}
It is worth mentioning that $C(\theta_0)$ is independent at all the channel coefficients. Thus, $R_\text{b}^\dagger$ can be rewritten as follows:

\begin{equation}
	R_\text{b}^\dagger = \mathbb{E}[I(e, \mathbf{y})] = C(\theta_0) - \mathbb{E}_{\{\mathbf{y}_0\}}[H(e|\mathbf{y}_0)].
\end{equation}
The posterior probability of $e$ when receiving $\mathbf{y}_0$ is formulated as follows:

\begin{equation}
	p(e=j|\mathbf{y}_0) = \frac{\theta_jp(\mathbf{y}|e=j)}{\theta_0p(\mathbf{y}_0|e=0) +\theta_1 p(\mathbf{y}_0|e=1)},
\end{equation}
with $j \in \{0,1\}$. We then define $\omega_j = p(e=j|\mathbf{y}_0)$ with $j \in \{0,1\}$. Given the above, we can derive $H(e|\mathbf{y}_0)$ as follows:

\begin{equation}
	H(e|\mathbf{y}_0) = -\sum_{j=0}^{1}\omega_j\log_2\omega_j = C(\omega_0).
\end{equation}
Finally, the maximum achievable backscatter rate is obtained as follows:

\begin{equation}
	\begin{aligned}
		R_\text{b}^\dagger &= C(\theta_0) - \mathbb{E}_{\{\mathbf{y}_0\}}[C(\omega)]\\
		& = C(\theta_0) - \int_{\mathbf{y}_0}(\theta_0p(\mathbf{y}_0|e=0) +\theta_1 p(\mathbf{y}_0|e=1))C(\omega_0)d\mathbf{y}_0.
	\end{aligned}
\end{equation}

\end{document}